\documentclass{llncs}

\usepackage{graphics}
\usepackage[dvips]{epsfig}
\usepackage{breakcites}

\usepackage{amsmath}
\usepackage{amssymb}
\usepackage{amsfonts}
\usepackage{graphicx}
\usepackage{algorithm, algorithmic}

\newcommand{\classletter}{\ensuremath{\mathcal{G}}}
\newcommand{\ourclass}{\ensuremath{\classletter (k,\ell)}}
\newcommand{\treasure}{\ensuremath{g_{\ell+1}}}

\newcommand{\remove}[1]{}

\begin{document}

\title{{\bf Fast Rendezvous with Advice}}
\date{}

\author{
Avery Miller,
Andrzej Pelc\thanks{Partially supported by NSERC discovery grant and by the Research Chair in Distributed Computing at the Universit\'e du Qu\'{e}bec en Outaouais.}\\
E-mails: \email{avery@averymiller.ca}, \email{pelc@uqo.ca}\\
}
\institute{Universit\'{e} du Qu\'{e}bec en Outaouais, Gatineau, Canada.}

\date{ }
\maketitle

\begin{abstract}

Two mobile agents (robots), starting from different nodes of an $n$-node network at possibly different times, have to meet at the same node.
This problem is known as {\em rendezvous}.
Agents move in synchronous rounds using a deterministic algorithm.
In each round, an agent decides to either remain
idle or to move to one of the adjacent nodes. 
Each agent has a distinct integer label from the set $\{1,\dots,L\}$, which it can
use in the execution of the algorithm, but it does not know the label of the other agent. 

The main efficiency measure of a rendezvous algorithm's performance is its {\em time} , i.e., the number of rounds from the start of the later agent until the meeting.
If $D$ is the distance between the initial positions of the agents, then $\Omega(D)$ is an obvious lower bound on the time of rendezvous. However, if each agent has no initial  knowledge other than its label,
 time $O(D)$ is usually impossible to achieve. We study the minimum amount of information that has to be available  {\em a priori} to the agents to achieve rendezvous in optimal time $\Theta(D)$.
Following the standard paradigm of {\em algorithms
with advice}, this information is provided to the agents at the start by an oracle knowing the entire instance of the problem, i.e., the network, the starting  positions of the agents, their wake-up rounds,  and both of their labels. The oracle
helps the agents by providing them with the {\em same} binary string called {\em advice}, which can be used by the agents during their navigation. The length of this
string is called the {\em size of advice}.  Our goal is to find the smallest size of advice which enables the agents to meet in time $\Theta(D)$. We show that this optimal size of advice is $\Theta(D\log(n/D)+\log\log L)$. 
The upper bound is proved by constructing an advice string of this size, and providing a natural rendezvous algorithm using this advice that works in time $\Theta(D)$
for all networks.
The matching lower bound, which is the main contribution of this paper, is proved by exhibiting classes of networks for which it is impossible to achieve
rendezvous in time $\Theta(D)$
with smaller advice.

\vspace{2ex}

\end{abstract}

\vspace{2ex}

\vfill

\thispagestyle{empty}
\pagebreak

%%%%%%%%%%%%%%%%%%%%%%%%%%%%%%%%%%%%%%%%%%%%%%%%%%%%%%%%%%%
\section{Introduction}
%%%%%%%%%%%%%%%%%%%%%%%%%%%%%%%%%%%%%%%%%%%%%%%%%%%%%%%%%%%

\subsection{Background}

Two mobile agents, starting from different nodes of a network, have to meet at the same node at the same time.
This distributed task is known as {\em rendezvous} and has received a lot of attention in the literature.
Agents can be any mobile autonomous entities. They might represent human-made objects, such as software agents in computer networks or mobile robots navigating in a network of corridors in a building or a mine. They might also be natural, such as animals seeking a mate,
or people who want to meet in an unknown city whose streets form a network. 
The purpose of meeting in the case of software agents or mobile robots might be the exchange of data previously collected by the agents
or samples collected by the robots. It may also be the coordination
of future network maintenance tasks, for example checking functionality of websites or of sensors forming a network, or decontaminating corridors of a mine.

\subsection{Model and Problem Description}

The network is modeled as an undirected connected graph with $n$ unlabeled nodes.
We seek deterministic rendezvous algorithms that do not
rely on the agents perceiving node identifiers, and therefore can work in anonymous graphs as well  (cf. \cite{alpern02b}). 
The reason for designing such algorithms
is that, even when nodes of the network have distinct identifiers, agents may be unable to perceive them
because of limited sensory capabilities (e.g., a mobile robot may be unable to read signs at corridor crossings), 
or nodes may be unwilling to reveal their identifiers to software agents, e.g., due to security or privacy reasons.
From a methodological point of view, if nodes had distinct identifiers visible to the agents, the agents could explore the graph and meet at the node
with the smallest identifier. Hence, in this case, rendezvous
reduces to graph exploration.

On the other hand, we assume that, at each node $v$,
each edge incident to $v$ has a distinct {\em port number} from the set 
$\{0,\dots,d-1\}$, where $d$ is the degree of $v$. These port numbers are fixed and visible to the agents.
Port numbering is {\em local} to each node, i.e., we do not assume any relation between
port numbers at  the two endpoints of an edge. Note that in the absence of port numbers, edges incident to a node
would be undistinguishable for agents and thus rendezvous would be often impossible, 
as an agent may always miss some particular edge incident to the current node, and this edge could be a bridge to the part of the graph
where the other agent started.
The previously mentioned security and privacy reasons for not revealing node identifiers to software agents are irrelevant in the case of port numbers. If
the graph models a system of corridors of a mine or a building, 
port numbers can be made implicit, e.g., by marking one edge at each intersection
(using a simple mark legible even by a mobile robot with very limited vision),
considering it as corresponding to port 0, and all other port numbers increasing clockwise.

Agents are initially located at different nodes of the graph and  traverse its edges in synchronous rounds.
They cannot mark visited nodes or traversed edges in any way, and they cannot communicate before meeting.
The adversary wakes up each of the agents, possibly in different rounds. 
%A dormant agent is woken up by the first agent that visits its starting node. 
Each agent starts executing the algorithm in the round of its wake-up.
It has a clock that ticks at each round and starts at the wake-up round of the agent.
In each round, each agent either remains at the current node,
or chooses a port in order to move to one of the adjacent nodes. 
When an agent enters a node, it learns the node's degree and the port  number by which it enters the node. When agents cross each other
on an edge while traversing it simultaneously in different directions, they do not notice this fact.

Each agent has a distinct integer label from a fixed {\em label space} $\{1,\dots,L\}$, which it can
use as a parameter in the same deterministic algorithm that both agents execute. It does not know the label nor the starting round of the other agent. 
Since we study deterministic rendezvous, the absence of distinct labels would preclude the possibility of meeting in highly
symmetric graphs, such as rings or tori, for which there exist non-trivial port-preserving automorphisms. Indeed, in such graphs,
identical agents starting simultaneously and executing the same deterministic algorithm can never meet, since they will keep the same positive distance in every round. 
Hence, assigning different labels to agents is the only way to break symmetry, as is needed to meet in every graph using a deterministic algorithm. 
On the other hand, if agents knew
each other's identities, then the smaller-labelled agent could stay inert, while the other agent would try to find it. In this case rendezvous reduces to graph exploration.   
Assuming such knowledge, however, is not realistic, as agents are often created independently in different parts of the graph, and they know nothing about each other
prior to meeting.

%We assume that the memory of the agents is unlimited: from the computational point of view they are modeled as 
%Turing machines. 

The rendezvous is defined as both agents being at the same node in the same round.
The main efficiency measure of a rendezvous algorithm's performance is its {\em time} , i.e., the number of rounds from the start of the later agent until the meeting.
If $D$ is the distance between the initial positions of the agents, then $\Omega(D)$ is an obvious lower bound on the time of rendezvous. However, if the agents have no additional knowledge,
time $O(D)$ is usually impossible to achieve. This is due to two reasons. First, without any knowledge about the graph, even the easier task of {\em treasure hunt} \cite{TSZ07}, in
which a single agent must find a target (treasure) hidden at an unknown node of the graph, takes asymptotically larger time  in the worst case. Treasure hunt is equivalent to a special case of rendezvous  where one of the agents is inert. In the worst case, this takes as much time as graph exploration, i.e., having a single agent visit all nodes.
Second, even when the graph is so simple that navigation of the agents is not a problem, breaking symmetry between the agents, which is often necessary to achieve a meeting, may take time larger than $D$.
Indeed, even in the two-node graph, where $D=1$, rendezvous requires time $\Omega(\log L)$ \cite{DFKP}. 

We study the amount of information that has to be given  {\em a priori} to the agents to achieve rendezvous in optimal time $\Theta(D)$.
Following the paradigm of {\em algorithms
with advice}  \cite{AKM01,CFP,CFIKP,DP,EFKR,FGIP,FIP1,FIP2,FKL,FP,FPR,GPPR02,IKP,KKKP02,KKP05,SN,TZ05}, this information is provided to the agents at the start, by an oracle knowing the entire instance of the problem, i.e., the graph, the starting  positions of the agents, their wake-up rounds, and both of their labels. The oracle
helps the agents by providing them with the {\em same} binary string called {\em advice}, which can be used by each agent, together with its own label, during the execution of the algorithm. The length of this
string is called the {\em size of advice}.  Our goal is to find the smallest size of advice (up to constant factors) which enables the agents to meet in time $\Theta(D)$.
In other words we want to answer the question:

\begin{quotation}
What is the minimum information that permits the fastest possible rendezvous?
\end{quotation}

where both ``minimum'' and ``fastest'' are meant up to multiplicative constants. 

Notice that, since the advice given to both agents is identical, it could not help break symmetry if agents did not have distinct labels.
Hence, even with large advice, the absence of distinct labels would preclude rendezvous in highly symmetric networks, as argued above.
Using the framework of advice permits us to quantify the amount of information
needed for an efficient solution of a given network problem (in our case, rendezvous), regardless of the type of information that is provided.

%--------------------------------------------------
\subsection{Our Results}
\label{subsec:ourresults}

For agents with labels from the set $\{1,\dots,L\}$,
we show that, in order to meet in optimal time $\Theta(D)$ in $n$-node networks, the minimum size of advice that has to be provided to the agents is  $\Theta(D\log(n/D)+\log\log L)$. 
The upper bound is proved by constructing an advice string of this size, and providing a natural rendezvous algorithm using this advice that works in time $\Theta(D)$
for all networks.
The matching lower bound, which is the main contribution of this paper, is proved by exhibiting classes of networks for which it is impossible to achieve rendezvous in time $\Theta(D)$
with smaller advice. 

Our algorithm works for arbitrary starting times of the agents, and our lower bound is valid even for simultaneous start.
As far as the memory of the agents is concerned, our algorithm has very modest requirements: an agent must only be able to store the advice and its own label. Hence
memory of size $\Theta(D\log(n/D)+\log L)$ is sufficient. On the other hand, our lower bound on the size of advice holds even for agents with unlimited memory.

%--------------------------------------------------

%--------------------------------------------------
\subsection{Related Work}
\label{subsec:relatwork}
%--------------------------------------------------

The problem of rendezvous has been studied both under randomized and deterministic scenarios.
An extensive survey of  randomized rendezvous in various models  can be found in
\cite{alpern02b}, cf. also  \cite{alpern95a,alpern02a,anderson90,baston98}. 
Deterministic rendezvous in networks has been surveyed in \cite{Pe}.
Several authors
considered geometric scenarios (rendezvous in an interval of the real line, e.g.,  \cite{baston98,baston01},
or in the plane, e.g., \cite{anderson98a,anderson98b}).
Gathering more than two agents was studied, e.g., in \cite{fpsw,lim96,thomas92}.

For the deterministic setting, many authors studied the feasibility and time complexity of rendezvous. For instance, deterministic rendezvous of agents that are equipped with tokens used to mark nodes was considered, e.g., in~\cite{KKSS}. Most relevant to our work are the results about 
deterministic rendezvous in arbitrary graphs, when the two agents cannot mark nodes, but have unique labels  \cite{DFKP,KM,TSZ07}.
In \cite{DFKP}, the authors present a rendezvous algorithm whose running time is polynomial in the size of the graph, in the length of the shorter
label and in the delay between the starting times of the agents. In \cite{KM,TSZ07}, rendezvous time is polynomial in the first two of these parameters and independent of the delay.

Memory required by the agents to achieve deterministic rendezvous was studied in \cite{FP2} for trees and in  \cite{CKP} for general graphs.
Memory needed for randomized rendezvous in the ring is discussed, e.g., in~\cite{KKPM08}. 

Apart from the synchronous model used in this paper, several authors investigated asynchronous rendezvous in the plane \cite{CFPS,fpsw} and in network environments
\cite{BCGIL,CLP,DGKKP,DPV}.
In the latter scenario, the agent chooses the edge to traverse, but the adversary controls the speed of the agent. Under this assumption, rendezvous
at a node cannot be guaranteed even in very simple graphs. Hence the rendezvous requirement is relaxed to permit the agents to meet inside an edge.

Providing nodes or agents with arbitrary kinds of information that can be used to perform network tasks more efficiently has been
proposed in \cite{AKM01,CFP,CFIKP,DP,EFKR,FGIP,FIP1,FIP2,FKL,FP,FPR,GPPR02,IKP,KKKP02,KKP05,SN,TZ05}. This approach was referred to as
{\em algorithms with advice}.  
Advice is given either to nodes of the network or to mobile agents performing some network task.
In the first case, instead of advice, the term {\em informative labeling schemes} is sometimes used.
Several authors studied the minimum size of advice required to solve
network problems in an efficient way. 
%Thus the framework of advice permits to quantify the amount of information
%needed for an efficient solution of a given network problem, regardless of the type of information that is provided.

 In \cite{KKP05}, given a distributed representation of a solution for a problem,
the authors investigated the number of bits of communication needed to verify the legality of the represented solution.
In \cite{FIP1} the authors compared the minimum size of advice required to
solve two information dissemination problems using a linear number of messages. 
In \cite{FKL} it was shown that advice of constant size given to the nodes enables the distributed construction of a minimum
spanning tree in logarithmic time. 
In \cite{EFKR} the advice paradigm was used for online problems.
In \cite{FGIP} the authors established lower bounds on the size of advice 
needed to beat time $\Theta(\log^*n)$
for 3-coloring cycles and to achieve time $\Theta(\log^*n)$ for 3-coloring unoriented trees.  
In the case of \cite{SN} the issue was not efficiency but feasibility: it
was shown that $\Theta(n\log n)$ is the minimum size of advice
required to perform monotone connected graph clearing.
In \cite{IKP} the authors studied radio networks for
which it is possible to perform centralized broadcasting in constant time. They proved that constant time is achievable with
$O(n)$ bits of advice in such networks, while
$o(n)$ bits are not enough. In \cite{FPR} the authors studied the problem of topology recognition with advice given to nodes. 
In \cite{DP} the task of drawing an isomorphic map was executed by an agent in a graph and the problem was to determine the minimum advice that has to be given to the agent
for the task to be feasible.

Among the papers using the paradigm of advice, \cite{CFIKP,FIP2} are closest to the present work, as they both concern the task of graph exploration by an agent.
In \cite{CFIKP} the authors investigated the minimum size of advice that has to be given to unlabeled nodes (and not to the agent)
to permit graph exploration by an agent modeled as a $k$-state automaton.
In \cite{FIP2} the authors
established the size of advice that has to be given to an agent in order to explore trees while obtaining 
competitive ratio better than 2.
To the best of our knowledge, rendezvous with advice has never been studied before.

\section{Algorithm and Advice}

Consider any $n$ node graph, and suppose that the distance between the initial positions of the agents is  $D$.
In this section, we construct an advice string of length $O(D\log(n/D)+\log\log L)$ and a rendezvous algorithm which achieves time $D$ using this advice.
We first describe the advice string. Let $G$ be the underlying graph and let $\ell_1$ and $\ell_2$ be the distinct labels of the agents, both belonging to the  label space $\{1,\dots,L\}$. Call the agent with label $\ell_1$ the first agent and the agent with label $\ell_2$ the second agent.
Let $x$ be the smallest index such that the binary representations of $\ell_1$ and $\ell_2$ differ on the $x$th bit. Without loss of generality assume that the $x$th bit is 0
in the binary representation of  $\ell_1$ and 1 in the  binary representation of $\ell_2$.

Let $P$ be a fixed shortest path in $G$ between the initial positions $u$ and $v$ of the agents.
The path $P$ induces two sequences of ports of length $D$: the sequence $\pi'$ of consecutive ports to
be taken at each node of path $P$ to get from $u$ to $v$, and the sequence $\pi''$ of consecutive ports to
be taken at each node of path $P$ to get from $v$ to $u$. Let $\pi\in \{\pi',\pi''\}$ be the sequence 
corresponding to the direction from the initial position of the second agent to the initial position of the first agent.
Denote $\pi=(p_1,\dots,p_D)$. Let $A_i$, for $i=1,\dots, D$, be the binary representation of the integer $p_i$. Additionally,
let $A_0$ be the binary representation of the integer $x$. The binary strings  $(A_0,\ldots,A_{D})$ will be called substrings.
 
The sequence of substrings $(A_0,\ldots,A_{D})$ is encoded into a single advice string to pass to the algorithm. More specifically, the sequence is encoded by doubling each digit in each substring and putting 01 between substrings. This permits the agent to unambiguously decode the original sequence. Denote by $Concat(A_0,\ldots,A_{D})$ this encoding and let $Decode$ be the inverse (decoding) function, i.e., $Decode(Concat(A_0,\ldots,A_{D}))=(A_0,\ldots,A_{D})$.
As an example, $Concat((01),(00)) = (0011010000)$. Note that the encoding increases the total number of advice bits by a constant factor. The advice string 
given to the agents is ${\cal A}= Concat(A_0,\ldots,A_{D})$.

The idea of the Algorithm {\tt Fast Rendezvous} using the advice string $\cal A$  is the following. Each agent decodes the sequence $(A_0,\ldots,A_{D})$ from the string
$\cal A$. Then each agent looks at the $x$th bit of its label, where $x$ is the integer represented by $A_0$. If this bit is 0, the agent stays inert at its initial position, otherwise it takes the 
consecutive ports $p_1,\dots,p_D$, where $p_i$, for $i=1,\dots, D$, is the integer with binary representation $A_i$. After these $D$ moves, the agent meets the other agent
at the latter's initial position. Below is the pseudocode of the algorithm.

\begin{center}
\fbox{
\begin{minipage}{0.8\columnwidth}\small

{\bf Algorithm} {\tt Fast Rendezvous}\\

Input: advice string $\cal A$, label $\ell$.\\

$(A_0,\ldots,A_{D}):= Decode({\cal A})$\\
$x:=$ the integer with binary representation $A_0$.\\
{\bf if} the $x$th bit of $\ell$ is 1 {\bf then}\\
\hspace*{1cm}{\bf for} $i=1$ {\bf to} $D$ {\bf  do}\\
\hspace*{2cm}$p_i:=$ the integer with binary representation $A_i$\\
\hspace*{2cm}take port $p_i$\\
stop.

\end{minipage}
}
\end{center}

\begin{theorem}\label{ub}
Let $G$ be an $n$-node graph with two agents initially situated at distance $D$ from one another.
Algorithm {\tt Fast Rendezvous} achieves rendezvous in time $D$, using advice of size
 $O(D\log(n/D)+\log\log L)$.
\end{theorem}

\begin{proof}
The correctness and time of the algorithm are straightforward. It remains to prove that the
 length of the advice string $\cal A$ is $O(D\log(n/D)+\log\log L)$. To do this, it is enough to show that the sum of the lengths $z_i$ of substrings $A_i$, for $i=0,\dots, D$,
 is $O(D\log(n/D)+\log\log L)$. 
 
 Note that $z_0=\lceil \log x \rceil \in O(\log\log \ell) \subseteq O( \log\log L)$. 
 Also, note that
 $p_i \leq deg(v_i)$, where $v_1,\dots, v_D$ are consecutive nodes on some shortest path between the initial positions of the agents.
 It is well-known (cf. e.g.,  \cite{FP1}) that the sum of degrees on a shortest path between any two nodes in an $n$-node graph is bounded above
 by $3n$. Hence $p_1+\dots +p_D \leq 3n$. We have 
 $$z_1+\dots +z_D \leq  \lceil \log deg(v_1) \rceil+\dots + \lceil \log deg(v_D) \rceil \leq D+ \log(\prod_{i=1}^{D}deg(v_i)).$$
 The product of $D$ positive numbers whose sum is at most $3n$ is maximized when they are all equal to $3n/D$.
 Hence $z_1+\dots +z_D \leq D + D\log (3n/D) \in O(D\log (n/D))$. This concludes the proof.
\end{proof}

\section{Lower Bound}

In this section, we prove a lower bound on the size of advice permitting rendezvous in optimal time $O(D)$, where $D$ is the initial distance between the agents. This lower bound will match the upper bound established in
Theorem \ref{ub}, which says that, for an arbitrary $n$-node graph, rendezvous can be achieved in time $O(D)$ using advice of size
 $O(D\log(n/D)+\log\log L)$. In order to prove that this size of advice cannot be improved in general, we present two classes of graphs: one that requires
 advice $\Omega(D\log(n/D))$ and another that requires advice $\Omega(\log\log L)$ to achieve optimal time of rendezvous. To make the lower bound even stronger, we show that it holds even in the scenario where agents start simultaneously. 
 
 The $\Omega(D\log(n/D))$  lower bound will be first proved for the simpler problem of treasure hunt. Recall that in this task,
 a single agent must find a stationary target (treasure) hidden at an unknown node of the graph at distance $D$ from the initial position of the agent. 
 %As mentioned in the introduction, this is equivalent to a special case of rendezvous, where one of the agents is inert. 
% Hence a lower bound on the size of advice needed to perform treasure hunt in time $O(D)$ implies the same lower 
 %bound on the size of advice needed to perform rendezvous in time $O(D)$.
 We then show how to derive the same lower bound on the size of advice for the rendezvous problem.
 
 The following technical lemma gives a construction of a graph which will provide the core of our argument for the $\Omega(D\log(n/D))$ lower bound.

\begin{lemma}\label{lowerboundgraph}
Let $n$ and $D$ be positive integers such that $D\leq n/2$.
Consider any treasure-hunting algorithm $A$ that takes $Dz$ bits of advice. For any fixed even integer $k \in \{2,\ldots,n-1\}$ and every integer $\ell \in \{1, \ldots, \min\{\left\lfloor\frac{D}{2}\right\rfloor,\left\lfloor\frac{n-1}{k}\right\rfloor\}\}$, there exists a graph of size $k\ell+1$, an initial position of the agent in this graph,  and a location of the treasure at distance $2\ell$ from this initial position, for which algorithm $A$ uses $\Omega(\frac{\ell k^2}{2^{2z}})$ rounds.
\end{lemma}
\begin{proof}
We define a class of graphs $\ourclass$ such that each graph in $\ourclass$ has $k\ell + 1$ nodes. We will prove that there is a non-empty subset $B$ of $\ourclass$ such that, on each graph in $B$, algorithm $A$ uses $\Omega(\frac{\ell k^2}{2^{2z}})$ rounds to complete treasure hunt, for some initial position of the agent and a location of the 
treasure at distance $2\ell$ from this location. 

Each graph $G$ in the class consists of $\ell$ copies of a $k$-clique $H$ 
(with a port numbering to be described shortly), which are chained together in a special way. We will refer to these cliques as $H_1,\ldots,H_\ell$. 

Let $v_1,\ldots,v_k$ denote the nodes of $H$. It should be stressed that names of nodes in cliques are for the convenience of the description only, and they are not visible to the agent. 
We choose an arbitrary edge-colouring of $H$ using the colours $\{0,\ldots,k-2\}$, which is always possible for cliques of even size, cf. e.g., \cite{Gib}, Theorem 7.6. For an arbitrary edge $e$ in $H$, let $c(e)$ denote the colour assigned to $e$. The port numbers of $H$ are simply the edge colours, i.e., for any edge $\{u,v\}$, the port numbers corresponding to this edge at $u$ and $v$ are both equal to $c(\{u,v\})$. 

Each graph $G \in \ourclass$ is obtained by chaining together the copies $H_1,\ldots,H_\ell$ of the clique $H$ in the following way. We will call node $v_1$ in clique $H_i$ the \emph{gate} $g_i$ of $H_i$. The initial position of the agent is $g_1$. Each gate $g_i$, for $i > 1$, is placed on (i.e., subdivides) one of the edges of clique $H_{i-1}$ not incident to $g_{i-1}$. We denote this edge by $e_{i-1}$. Finally, an additional \emph{treasure node} $\treasure$ is placed on  (i.e., subdivides) one of the edges of clique $H_\ell$ not incident to $g_{\ell}$, and this edge is denoted by $e_\ell$.
Hence $g_1$ has degree $k-1$, each $g_i$, for $1<i\leq \ell$, has degree $k+1$, and $g_{\ell +1}$ has degree 2, cf. Figure 1 (a).
Note that, since $g_i$, for $i > 1$, subdivides an edge that is not incident to $g_{i-1}$, we have $D=2\ell$.  
Port numbering of graph $G$ is the following. Port numbers in each clique $H_{i}$ are
unchanged, the new port numbers at each node $g_{i}$, for $1< i  \leq \ell$, are $k-1$ and $k$, with $k-1$ corresponding to the edge whose other endpoint has smaller index,
and the new port numbers at node $g_{\ell +1}$ are $0$ and $1$, with $0$ corresponding to the edge whose other endpoint has smaller index,
cf. Figure 1 (b).
All graphs in the class $\ourclass$ are isomorphic and differ only by port numbering. 
Note that each graph in $\ourclass$ is uniquely identified by the sequence of edges $(e_1,\ldots,e_\ell)$. Therefore, the number of graphs in $\ourclass$ is 
$N = ((k-1)(k-2)/2)^\ell$.

\begin{figure}[!ht]
\begin{center}
\includegraphics[scale=0.6]{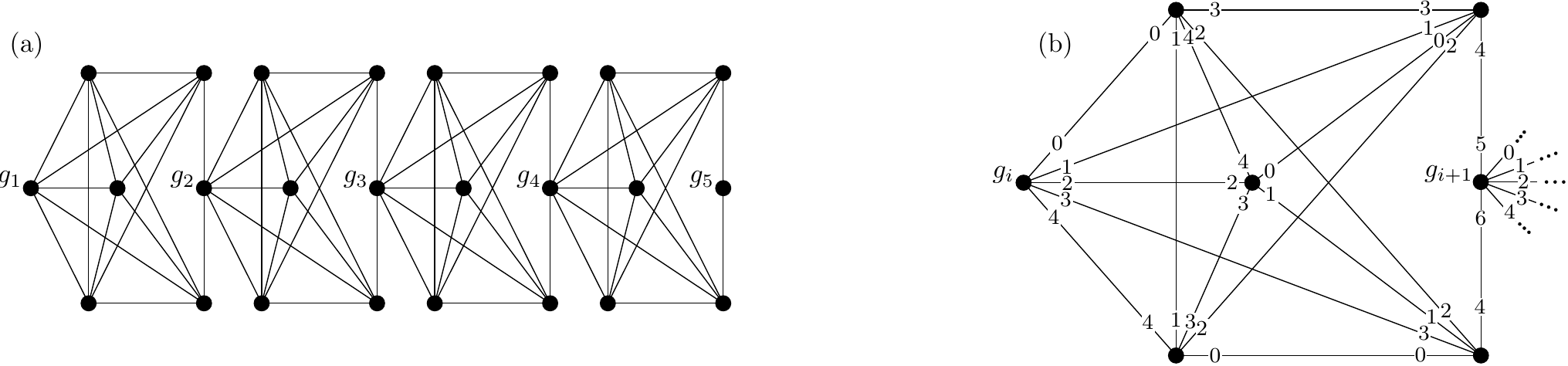}
\end{center}
\caption{(a) A graph in $\classletter (6,4)$ (b) Port numbering of each clique $H_i$, for $i<\ell$, with gate $g_{i+1}$ inserted }
\label{GraphDiagrams}
\end{figure}

Notice that an agent navigating in a graph $G \in \ourclass$  always knows when it arrives at a gate $g_i$, for $1<i\leq \ell$, because these are the only nodes of degree
$k+1$. Call a walk of an agent {\em normal}, if the agent visits each gate $g_i$, for $1<i\leq \ell$, exactly once (i.e., never exits a gate by port $k-1$ or $k$).
It is enough to prove our lower bound on the time of treasure hunt only for algorithms where the agent always performs a normal walk. Indeed,  for any walk, there exists a normal walk using at most the same time. From now on we restrict attention to such algorithms.

We prove our lower bound on the class of graphs $\ourclass$. The idea is that, in order to find the treasure node, the agent must visit each of the nodes $g_1,\ldots,\treasure$. To get from $g_i$ to $g_{i+1}$, the agent must find the edge $e_i$ of $H_i$ that the node $g_{i+1}$ subdivides. With little advice, this amounts to searching many edges of the clique $H_i$, and hence increases time.

For any graph $G$, the agent is given some advice string $S$ and executes its treasure-hunting algorithm $A$. 
With $Dz$ bits of advice, there exists a set $B$ of at least $\frac{N}{2^{Dz}}$ graphs for which the agent is given the same advice string. Next, we provide an upper bound on the number of graphs in $B$.  By comparing this upper bound with $\frac{N}{2^{Dz}}$, we will get the desired lower bound on the number of rounds needed to find the treasure.
%Note that, with $\lceil \log N \rceil$ bits of advice, the agent can be told exactly which instance it is running on, so it can find the treasure in $2\ell \in O(D)$ rounds by following the port numbers $(a_1,b_1,a_2,b_2,\ldots,a_\ell,b_\ell)$. However, with fewer bits of advice, it follows from the Pigeonhole Principle that there are two or more instances for which the agent will be given the same advice. 

Let $T$ be the maximum running time of algorithm $A$ on  graphs of class $\ourclass$. 
Let $\tau$ be the function that maps each graph from $B \subseteq \ourclass$ to an $\ell$-tuple $(t_1,\ldots,t_\ell)$, where, for each $i \in \{1,\ldots,\ell\}$, $t_i$ is the number of edge traversals performed by the agent in clique $H_i$. This function is well-defined since we consider only deterministic algorithms. We now prove that this function is injective.

\begin{claim}\label{injective}
For any two graphs $G \neq G'$ in the set $B$, we have $\tau(G) \neq \tau(G')$.
\end{claim}
\begin{proof}
Let $G$ and $G'$ be represented by the  distinct sequences of edges $(e_1,\ldots,e_\ell)$ and $(e'_1,\ldots,e'_\ell)$, respectively.
Let $j$ be the first index for which $e_j\neq e'_j$. Since the advice string for graphs $G$ and $G'$ is the same, the sequence of ports taken by the agent in these graphs
is the same until the agent reaches nodes of different degree in $G$ and $G'$. Let  $\tau(G)=(t_1,\ldots,t_\ell)$ and  $\tau(G')=(t'_1,\ldots,t'_\ell)$.
Hence $t_i=t'_i$ for $i<j$. Without loss of generality assume that $t_j \leq t'_j$. For each round $r\leq t_1+\dots + t_j$, the agent  takes the same port number in $G$ and $G'$. In round $t_1+\dots + t_j+1$ the agent reaches the gate $g_{j+1}$ in $G$. Since $e_j\neq e'_j$, the agent does not reach the gate  $g_{j+1}$ in $G'$ in this round. Since the walk of the agent is normal, we conclude that
$t_j<t'_j$.  
\end{proof}

By Claim \ref{injective}, the number of graphs in $B$ is bounded above by the size of the range of $\tau$. Consider an arbitrary $\ell$-tuple $(t_1,\ldots,t_\ell)$ in the range of $\tau$. By the definition of $\ourclass$, for each $i \in \{1,\ldots,\ell\}$, the agent must traverse at least two edges to get from $g_i$ to $g_{i+1}$. Further, $T$ is an upper bound on the number of edge traversals performed in any execution of the algorithm. Therefore, the size of the range of $\tau$ is bounded above by the number of integer-valued $\ell$-tuples with positive terms whose sum is at most $T$.  Clearly, this is bounded above by the number of real-valued $\ell$-tuples with non-negative terms whose sum is at most $T$, i.e., by the size of the simplex $\Delta_\ell = \{ (t_1,\ldots,t_\ell) \in \mathbb{R}^{\ell}\ |\ \sum\limits_{i=1}^{\ell} t_i = T \textrm{ and } 0 \leq t_i \leq T \textrm{ for all $i$} \}$. From \cite{ellis}, the volume of $\Delta_\ell$ is equal to $T^\ell/\ell !$. Thus, we have shown that the  size of $B$ is bounded above by $T^\ell/\ell !$. Comparing this to our lower bound $\frac{N}{2^{Dz}}$ on the size of $B$, we get
\[
T  \geq   \sqrt[\ell]{\ell ! \frac{N}{2^{Dz}}} 
  \geq   \sqrt[\ell]{\ell ! \frac{((k-1)(k-2)/2)^{\ell}}{2^{Dz}}} 
   \geq   \sqrt[\ell]{\ell !} \frac{(k-2)^{2}/2}{2^{Dz/\ell}}\\
  =  \sqrt[\ell]{\ell !} \frac{(k-2)^{2}/2}{2^{2z}}.
\]

By Stirling's formula we have $\ell ! \geq \sqrt{\ell}(\ell/e)^{\ell}$, for sufficiently large $\ell$. Hence $\sqrt[\ell]{\ell !}\geq  \ell^{1/(2\ell)} \cdot (\ell/e)$.
Since the first factor converges to 1 as $\ell$ grows, we have $\sqrt[\ell]{\ell !} \in \Omega(\ell)$. Hence the above bound on $T$ implies $T \in \Omega(\frac{\ell k^{2}}{2^{2z}})$.
\end{proof}

\begin{theorem}\label{lb1}
Let $n$ and $D$ be positive integers such that $D\leq n/2$.
 If an algorithm  $A$ solves treasure hunting in $O(D)$ rounds whenever the treasure is at distance $D$ from the initial position of the agent, 
 then there exists an $n$-node graph $G$ with treasure at this distance such that $A$ requires $\Omega(D\log(n/D))$ bits of advice.
\end{theorem}
\begin{proof}
We suppose that only $o(D\log(n/D))$ bits of advice are provided, and show that there is an $n$-node graph for which $A$ completes treasure hunt in $\omega(D)$ rounds.

{\bf Case 1.} $D \in o(n)$.\\
The amount of advice provided is $Dz$, where $z < \frac{1}{2}\log(n/D)$. By Lemma \ref{lowerboundgraph} with $k =  \lfloor n/D \rfloor$ or $k= \lfloor n/D \rfloor-1$, whichever number is even, and $\ell = \lfloor (D-1)/2 \rfloor$, there exists a graph $G' \in \ourclass$ such that, for some positive constant $c$, algorithm $A$ uses at least $c\frac{(D/2)(n/D)^2}{2^{2z}} > \frac{c}{2}\frac{n^2/D}{2^{\log(n/D)}} = \frac{cn}{2} \in \omega(D)$ rounds to reach the treasure located at some node $x$ at distance $2\ell$
from the initial position of the agent. 

{\bf Case 2.} $D \in \Omega(n)$.\\
Since $\log (n/D)$ is constant and $o(D\log(n/D))$ bits of advice are provided, it follows that
the amount of advice provided is $Dz$, where $z = 1/\varphi(n)$ for some integer function $\varphi(n)$ such that $\varphi(n) \in \omega(1)$. By Lemma \ref{lowerboundgraph} with $k = \min \{n-D-1, \varphi(n) \}$ or $k=\min \{n-D-1, \varphi(n) \}-1$, whichever number is even, and $\ell = \left\lfloor \frac{n-D-1}{k} \right\rfloor$, there exists a graph $G' \in \ourclass$ such that, for some positive constant $c$, algorithm $A$ uses at least $c\frac{((n-D-1)/k)(k^2)}{2^{2z}} > \frac{c}{4}(n-D-1)k \in \omega(n)$ rounds to reach the treasure located at some node $x$ at distance $2\ell$
from the initial position of the agent.

In both cases, graph $G'$ has $k\ell + 1$ nodes and the treasure is located at distance $2\ell$ from the initial position of the agent.
We obtain a graph $G$ with $n$ nodes by attaching a path of $n-k\ell-1$ nodes  to node $x$ in $G'$. In this graph $G$ the initial position of the agent is unchanged
with respect to $G'$, and 
the treasure is located on the attached path at distance $D$ from this initial position. 
\end{proof}

We now show how to deduce a lower bound on the size of advice for rendezvous (even with simultaneous start) from the lower bound for treasure hunt.

\begin{corollary}\label{cor}
Let $D' \leq n'$ be positive integers. 
There exist $n \in \Theta(n')$ and $D \in \Theta(D')$ such that 
if an algorithm $A$ solves rendezvous in time $O(D)$ in $n$-node graphs
whenever the initial distance between the agents is $D$, then there exists an $n$-node graph for which $A$ requires
$\Omega(D\log(n/D))$ bits of advice.
\end{corollary}

\begin{proof}
Let $G$ be a $(2n')$-node graph with a treasure located at distance $D'$ from the initial position of the agent
such that an $O(D')$-time treasure hunting algorithm requires $\Omega(D'\log(n'/D'))$ bits of advice. Such a graph exists by
Theorem \ref{lb1}. Let $v$ be the initial position of the agent and let $w$ be the location of the treasure in graph $G$. 

Let $n=4n'$ and $D=2D'+1$. Hence $n \in \Theta(n')$ and $D \in \Theta(D')$.
Construct the graph $G^*$ which consists of two disjoint copies of $G$ with the treasure locations joined by an edge. Locate two agents in $G^*$, each at the node $v$ in different copies of $G$.
The graph $G^*$ has $n$ nodes and the initial distance between the agents is $D$. In order to accomplish rendezvous in time $O(D)$, at least one of the agents has to
traverse the edge joining the copies, hence it must find the node $w$ in its copy in time $O(D')$,  which requires advice of size $\Omega(D'\log(n'/D'))=\Omega(D\log(n/D))$.
\end{proof}

The second part of our lower bound on the size of advice, i.e., the lower bound $\Omega(\log\log L)$, will be proved on the class of oriented rings.
A ring is {\em oriented} if every edge has port labels 0 and 1 at the two end-points.
Such a port labeling induces orientation of the ring: at each node, we will say that taking port 0 is going clockwise and taking port 1 is going counterclockwise.
We assume that agents operate in an oriented ring of size $n$. In order to make the lower bound as strong as possible, we prove that it holds even for simultaneous start of the agents. 

\begin{theorem}\label{lb2}
Let $D' \leq n'$ be positive integers. Consider any algorithm $A$ that solves rendezvous for agents with labels from the set $\{1,\dots, ,L\}$. 
There exist $n \in \Theta(n')$ and $D \in \Theta(D')$ such that if $A$ uses time  $O(D)$ in the $n$-node oriented ring whenever the initial distance between the agents is $D$, then the required size of advice is
$\Omega(\log\log L)$.
\end{theorem}

\begin{proof}
Assume that $S$ is the advice string given to the agents. Consider an agent with 
label $x \in \{1,\ldots,L\}$ executing algorithm $A$ using advice $S$. The actions of the agent in consecutive rounds until rendezvous are specified by  
a \emph{behaviour vector} $V_x$. In particular, $V_x$ is a sequence with terms from $\{-1,0,1\}$ that specifies, for each round $i$, whether agent  $x$ moves clockwise (denoted by $-1$), remains idle (denoted by $0$), or moves counter-clockwise (denoted by $1$). Note that an agent's behaviour vector is independent of its starting position, since all nodes of the ring look the same to the agent.  This behaviour vector depends exclusively on the label of the agent and on the advice string $S$.

Let $D=3D'$, $m=n'-(n'\mod D')$ and $n=\max(m, 6D')$. Hence $n \in \Theta(n')$, $D \in \Theta(D')$, $D'$ divides $n$, and $n\geq 2D$. 
As the initial positions of the agents,
fix any nodes $v$ and $w$ of the $n$-node oriented ring, where $w$ is at clockwise distance $D$ from $v$. Since $n\geq 2D$,
agents are at distance $D$ in the ring. 
%Let $a=\lfloor D/4 \rfloor$. 
%For simplicity assume that $n$ is divisible by $a$. The argument is easy to modify in the general case. 
%Partition all nodes of the ring into consecutive blocks of size $a$. 
Partition the nodes of the ring into $r$ consecutive blocks $B_0,B_1,\dots, B_{r-1}$ of size $D'$, starting clockwise from node $v$.
Hence the initial positions $v$ and $w$ of the agents are the clockwise-first nodes of block $B_0$ and block $B_3$, respectively.
Since agents start simultaneously, we have the notion of global round numbers counted since their start.
Partition all rounds $1,2,\dots$ into consecutive {\em time segments} of length $D'$. Hence, during any time segment, an agent can be
located in at most two (neighbouring) blocks.

Fix a behaviour vector $V_x$ of an agent with label $x$.
We define its {\em meta-behaviour vector} as a sequence  $M_x$ with terms from $\{-1,0,1\}$ as follows. 
Suppose that the agent is in block $B_j$ in the first round of the $i$-th segment.
The $i$-th term of $M_x$ is $z \in \{-1,0,1\}$, if, in the first round of the $(i+1)$-th time segment, the agent is in the block $B_{j+z}$, where index addition is modulo $r$. 
Since the initial position of an agent  is the clockwise-first node of a block, for a fixed behaviour vector of an agent its meta-behaviour vector is well defined.

Suppose that algorithm $A$ takes at most $cD$ rounds, for some constant $c$. This corresponds to $d$ time segments for some constant $d\leq 3c$. Hence,
all meta-behaviour vectors describing the actions of agents before the meeting are of length $d$ (shorter meta-behaviour vectors can be padded by zeroes at the end.) Let $\cal B$ be the set of sequences of length $d$ with 
terms from $\{-1,0,1\}$. Sequences from $\cal B$ represent possible meta-behaviour vectors of the agents. The size of $\cal B$ is $3^d$. 

%Let the initial positions of the agents be the clockwise-first nodes of block $B_0$ and block $B_3$, respectively. Hence agents start at distance $D$.
Since the initial positions of the agents are in blocks that are separated by two other blocks, agents with the same
meta-behaviour vectors must be in different blocks in every round, and hence they can never meet. Indeed, in the first round of every time segment they must be in blocks separated by two other blocks, and during any time segment, an agent can either stay in the same block or get to an adjacent block.

Suppose that the number of bits of advice is at most
$\frac{1}{2}\log\log L$. It follows that the set $\cal A$ of advice strings is of size at most $\sqrt{\log L}$. For any label $x \in \{1,\ldots,L\}$, let 
$\Phi _x$ be the function from $\cal A$ to $\cal B$, whose value on an advice string $S \in \cal A$ is the meta-behaviour vector of the agent with label $x$ when given the advice string $S$. Functions $\Phi _x$ are well-defined, as the meta-behaviour vector of an agent whose initial position is the clockwise-first node of a block depends only on its behaviour vector, which in turn depends only on the label of the agent and on the advice string.

 If the set ${ \cal B}^{ \cal A}$ of all functions from ${\cal A}$ to ${ \cal B}$ had fewer elements than $L$, then there would exist two distinct labels $x_1$ and $x_2$ of agents
 such that, for any advice string $S$, these agents would have an identical meta-behaviour vector. As observed above, these agents could never meet. This implies
 $(3^d)^{\sqrt{\log L}}\geq |{ \cal B}^{ \cal A}|\geq L$. Hence $d \log 3 \geq \sqrt{\log L}$, which contradicts the fact that $d$ is a constant.
 
 This shows that the number of bits of advice must be larger than $\frac{1}{2}\log\log L \in \Omega(\log\log L)$.
\end{proof}

Corollary \ref{cor} and  Theorem \ref{lb2} imply:

\begin{theorem}\label{lb}
Let $D' \leq n'$ be positive integers.
Consider any algorithm $A$ that solves rendezvous for agents with labels from the set $\{1,\dots, ,L\}$. 
There exist $n \in \Theta(n')$ and $D \in \Theta(D')$ such that, if $A$ takes time $O(D)$ in all $n$-node graphs
whenever the initial distance between the agents is $D$, then the required size of advice is
$\Omega(D\log (n/D)+\log\log L)$.
\end{theorem}

Theorems \ref{ub} and \ref{lb} imply the following corollary which is the main result of this paper.

\begin{corollary}
The minimum number of bits of advice sufficient to accomplish rendezvous of agents with labels from the set $\{1,\dots ,L\}$ in all $n$-node graphs in time $O(D)$, whenever the initial distance between the agents is $D$, is $\Theta((D\log (n/D)+\log\log L))$.
\end{corollary}

\section{Conclusion}

We established that  $\Theta(D\log(n/D)+\log\log L)$ is the minimum amount of information (advice) that agents must have in order to meet in optimal time $\Theta(D)$, where $D$ is the initial distance between them.
It should be noted that the two summands in this optimal size of advice have very different roles. 
On one hand, $\Theta(D\log(n/D))$ bits of advice are necessary and sufficient
to accomplish, in $O(D)$ time, the easier task of treasure hunt in $n$-node networks, where a single agent must find a target (treasure) hidden at an unknown node of the network
at distance $D$ from its initial position. This task is equivalent to a special case of rendezvous where one of the agents is inert. On the other hand, for agents whose labels are drawn from a label space of size $L$,
$\Theta(\log\log L)$ bits of advice are needed to break symmetry quickly enough 
in order to solve rendezvous in time $O(D)$,  
and hence, are necessary to meet in optimal time $\Theta(D)$,  even in constant-size networks.
 It should be stressed that the first summand in $O(D\log(n/D)+\log\log L)$ is usually larger than the second. Indeed, only when $L$ is very large with respect to $n$ and $D$ does the second summand dominate. This means that  ``in most cases'' the easier task of solving treasure hunt in optimal time is as demanding,
 in terms of advice,  as the harder task of solving rendezvous in optimal time.
 
 In this paper, we assumed that the advice given to both agents is identical. How does the result change when each agent can get different advice? It is clear that
 giving only one bit of advice, 0 to one agent and 1 to the other, breaks symmetry between them, e.g., the algorithm can make the agent that received bit 0 stay inert. Thus, if advice can be different,
 one bit of advice reduces rendezvous to treasure hunt. The opposite reduction is straightforward. 
 Hence it follows from our results that $\Theta(D\log(n/D))$ bits of advice are necessary and sufficient
 to accomplish rendezvous in optimal time $\Theta(D)$ in $n$-node networks, if advice can be different. This holds regardless of the label space and is, in fact, 
 also true for anonymous (identical) agents.

%\pagebreak

%%%%%%%%%%%%%%%%%%%%%%%%%%%%%%%%%%%%%%%%%%%%%%%%%%%%%%%%%%%
\bibliographystyle{plain}

\begin{thebibliography}{12}
%%%%%%%%%%%%%%%%%%%%%%%%%%%%%%%%%%%%%%%%%%%%%%%%%%%%%%%%%%%

\bibitem{AKM01}
S.~Abiteboul, H.~Kaplan, T.~Milo, Compact labeling schemes for ancestor
queries, Proc. 12th Annual ACM-SIAM Symposium on Discrete
Algorithms (SODA 2001), 547--556.

\bibitem{alpern95a}
S. Alpern,
The rendezvous search problem,
SIAM J. on Control and Optimization 33 (1995), 673-683.


\bibitem{alpern02a}
S. Alpern,
Rendezvous search on labelled networks,
Naval Reaserch Logistics 49 (2002), 256-274.

\bibitem{alpern02b}
S. Alpern and S. Gal,
The theory of search games and rendezvous.
Int. Series in Operations research and Management Science,
Kluwer Academic Publisher, 2002.

%\bibitem{alpern99}
%J. Alpern, V. Baston, and S. Essegaier,
%Rendezvous search on a graph,
%Journal of Applied Probability 36 (1999), 223-231.


\bibitem{anderson90}
E. Anderson and R. Weber,
The rendezvous problem on discrete locations,
Journal of Applied Probability 28 (1990), 839-851.

\bibitem{anderson98a}
E. Anderson and S. Fekete,
Asymmetric rendezvous on the plane,
Proc. 14th Annual ACM Symp. on Computational Geometry (1998), 365-373.

\bibitem{anderson98b}
E. Anderson and S. Fekete,
Two-dimensional rendezvous search,
Operations Research 49 (2001), 107-118.


\bibitem{BCGIL}
E. Bampas, J. Czyzowicz, L. Gasieniec, D. Ilcinkas, A. Labourel, Almost optimal asynchronous rendezvous in infinite multidimensional grids,
Proc. 24th International Symposium on Distributed Computing (DISC 2010),  297-311.


\bibitem{baston98}
V. Baston and S. Gal,
Rendezvous on the line when the players' initial distance is
given by an unknown probability distribution,
SIAM J. on Control and Opt. 36 (1998), 1880-1889.

\bibitem{baston01}
V. Baston and S. Gal,
Rendezvous search when marks are left at the starting
points,
Naval Reaserch Logistics 48 (2001), 722-731.

\bibitem{CFP}
S. Caminiti, I. Finocchi, R. Petreschi,
Engineering tree labeling schemes: a case study on least common ancestor,
Proc. 16th Annual European Symposium on Algorithms (ESA 2008), 234--245.



\bibitem{CFPS}
M. Cieliebak, P. Flocchini, G. Prencipe, N. Santoro, 
Distributed computing by mobile robots: Gathering, SIAM J. Comput. 41 (2012), 829-879.

\bibitem{CFIKP}
R. Cohen, P. Fraigniaud, D. Ilcinkas, A. Korman, D. Peleg, 
Label-guided graph exploration by a finite automaton, 
ACM Transactions on Algorithms 4 (2008).



\bibitem{CKP}
J. Czyzowicz, A. Kosowski, A. Pelc, How to meet when you forget: Log-space rendezvous in arbitrary graphs, Distributed Computing 25 (2012), 165-178. 

\bibitem{CLP}
J. Czyzowicz, A. Labourel, A. Pelc, How to meet asynchronously (almost) everywhere, 
ACM Transactions on Algorithms 8 (2012), article 37. 





\bibitem{DGKKP}
G. De Marco, L. Gargano, E. Kranakis, D. Krizanc, A. Pelc, U. Vaccaro, 
Asynchronous deterministic rendezvous in graphs, 
Theoretical Computer Science 355 (2006), 315-326.

\bibitem{DP}
D. Dereniowski, A. Pelc, Drawing maps with advice,  Journal of Parallel and Distributed Computing 72 (2012), 132--143. 



\bibitem{DFKP}
A. Dessmark, P. Fraigniaud, D. Kowalski, A. Pelc.
Deterministic rendezvous in graphs.
Algorithmica 46 (2006), 69-96.


\bibitem{DPV}
Y. Dieudonn\'{e}, A. Pelc, V. Villain, How to meet asynchronously at polynomial cost, Proc. 32nd Annual ACM Symposium on Principles of Distributed Computing (PODC 2013), 92-99.

\bibitem{ellis}
R. Ellis,
Volume of an N-Simplex by Multiple Integration, 
Elemente der Mathematik 31 (1976), 57-59.

\bibitem{EFKR}
Y. Emek, P. Fraigniaud, A. Korman, A. Rosen, Online computation with advice, Theoretical Computer Science 412 (2011), 2642--2656.




\bibitem{fpsw}
P. Flocchini, G. Prencipe, N. Santoro, P. Widmayer,
Gathering of asynchronous robots with limited visibility, Theoretical Computer Science 337 (2005), 147-168.


%\bibitem{FP}
%P. Fraigniaud, A. Pelc, Deterministic rendezvous in trees with little memory, 
%Proc. 22nd International Symposium on Distributed Computing (DISC 2008),  242-256. 

\bibitem{FGIP}
P. Fraigniaud, C. Gavoille, D. Ilcinkas, A. Pelc, 
Distributed computing with advice: Information sensitivity of graph coloring, 
Distributed Computing 21 (2009), 395--403.

\bibitem{FIP1}
P. Fraigniaud, D. Ilcinkas, A. Pelc, 
Communication algorithms with advice, Journal of  Computer and System Sciences 76 (2010), 222--232.
%Oracle size: a new measure of difficulty for communication problems, 
%Proc. 25th Ann. ACM Symposium on Principles of Distributed Computing, (PODC 2006), 179--187.

\bibitem{FIP2}
P. Fraigniaud, D. Ilcinkas, A. Pelc, 
Tree exploration with advice, Information and Computation 206 (2008), 1276--1287.
%Tree exploration with an oracle, 
%Proc. 31st International Symposium on Mathematical Foundations of Computer Science, (MFCS 2006), LNCS 4162, 
%24--37.

\bibitem{FKL}
P. Fraigniaud, A. Korman, E. Lebhar,
Local MST computation with short advice,
Theory of Computing Systems 47 (2010), 920--933.

\bibitem{FP2}
P. Fraigniaud, A. Pelc, Delays induce an exponential memory gap for rendezvous in trees, ACM Transactions on Algorithms 9 (2013), article 17. 




\bibitem{FP}
E. Fusco, A. Pelc, Trade-offs between the size of advice and broadcasting time in trees, Algorithmica 60 (2011), 719--734. 

\bibitem{FP1}
E. Fusco, A. Pelc, Communication complexity of consensus in anonymous message passing systems, Proc. 15th International Conference on Principles of Distributed Systems (OPODIS 2011), 191-206. 

\bibitem{FPR}
E. Fusco, A. Pelc, R. Petreschi, Use knowledge to learn faster: Topology recognition with advice, Proc. 27th International Symposium on Distributed Computing (DISC 2013), 31-45.

\bibitem{GPPR02}
C.~Gavoille, D.~Peleg, S.~P\'{e}rennes, R.~Raz.
Distance labeling in graphs, 
Journal of Algorithms 53 (2004), 85-112.

\bibitem{Gib}
A. Gibbons, Algorithmic Graph Theory, Cambridge University press 1985.


%\bibitem{FP3}
%P. Fraigniaud, A. Pelc,  Decidability classes for mobile agents computing, Proc. 10th Latin American Theoretical Informatics Symposium (LATIN %2012), 362-374. 

%\bibitem{gal99}
%S. Gal,
%Rendezvous search on the line,
%Operations Research 47 (1999), 974-976.

\bibitem{IKP}
D. Ilcinkas, D. Kowalski, A. Pelc, 
Fast radio broadcasting with advice, 
 Theoretical Computer Science, 411 (2012),  1544--1557.

\bibitem{KKKP02}
M.~Katz, N.~Katz, A.~Korman, D.~Peleg, Labeling schemes for flow and
connectivity, 
SIAM Journal of  Computing 34 (2004), 23--40.
%Proceedings of the 13th annual ACM-SIAM Symposium on
%Discrete Algorithms (SODA 2002), 927--936.


\bibitem{KKP05}
A. Korman, S. Kutten, D. Peleg, Proof labeling schemes,
Distributed Computing 22 (2010), 215--233.  




\bibitem{KM}
D. Kowalski, A. Malinowski,
How to meet in anonymous network,
Proc. 13th Int. Colloquium on Structural Information and Communication Complexity,
(SIROCCO 2006), 44-58. 

\bibitem{KKPM08}
E. Kranakis, D. Krizanc, and P. Morin, 
Randomized Rendez-Vous with Limited Memory,
Proc. 8th Latin American Theoretical Informatics (LATIN 2008), 605-616.

\bibitem{KKSS}
E. Kranakis, D. Krizanc, N. Santoro and C. Sawchuk, 
Mobile agent rendezvous in a ring, 
Proc. 23rd Int. Conference on Distributed Computing Systems
(ICDCS 2003), 592-599.

\bibitem{lim96}
W. Lim and S. Alpern,
Minimax rendezvous on the line,
SIAM J. on Control and Optimization 34 (1996), 1650-1665.

%\bibitem{Ly}
%N.L. Lynch, Distributed algorithms, Morgan Kaufmann Publ. Inc.,
%San Francisco, USA, 1996.

\bibitem{SN}
N. Nisse, D. Soguet, Graph searching with advice,
Theoretical Computer Science 410 (2009), 1307--1318.



\bibitem{Pe}
A. Pelc, Deterministic rendezvous in networks: A comprehensive survey, 
Networks 59 (2012), 331-347. 


%\bibitem{Re}
%O. Reingold, Undirected connectivity in log-space, Journal of the ACM 55 (2008).

\bibitem{TSZ07}
A. Ta-Shma and U. Zwick.
Deterministic rendezvous, treasure hunts and strongly universal exploration sequences.
Proc. 18th ACM-SIAM Symposium on Discrete Algorithms (SODA 2007), 599-608.

\bibitem{thomas92}
L. Thomas,
Finding your kids when they are lost,
Journal on Operational Res. Soc. 43 (1992), 637-639.

\bibitem{TZ05}
M.~Thorup, U.~Zwick, Approximate distance oracles,
Journal of the ACM, 52 (2005), 1--24.



%\bibitem{Y}
%A. C-C. Yao, Probabilistic computations: Towards a unified measure of 
%complexity, Proc. 18th Annual IEEE Conference on Foundations of Computer Science, (FOCS 1977), 222-227.



%%%%%%%%%%%%%%%%%%%%%%%%%%%%%%%%%%%%%%%%%%%%%%%%%%%%%%%%%%%
\end{thebibliography}

%%%%%%%%%%%%%%%%%%%%%%%%%%%%%%%%%%%%%%%%%%%%%%%%%%%%%%%%%%% 

\end{document}